\documentclass[]{article}
\usepackage{proceed2e,times}

\title{Combining predictions from linear models \\ when training and test
  inputs differ}

\author{ {\bf Thijs van Ommen} \\
  Centrum Wiskunde \& Informatica \\
  Amsterdam, The Netherlands
} %

\usepackage{amsmath}
\usepackage{amsfonts} %
\usepackage{latexsym}
\usepackage{graphicx}
\usepackage{amssymb} %
\usepackage[round]{natbib}
\DeclareRobustCommand{\VAN}[3]{#2}
\usepackage{etoolbox} %
\newtoggle{supplementary}
\toggletrue{supplementary}

\def\squarebox#1{\hbox to #1{\hfill\vbox to #1{\vfill}}}
\def\qed{\hspace*{\fill}
        \vbox{\hrule\hbox{\vrule\squarebox{.667em}\vrule}\hrule}}
\newenvironment{proof}[1][Proof]{\begin{trivlist}
    \item[\hskip \labelsep {\bfseries #1}]}{\qed \end{trivlist}}

\renewcommand{\eqref}[1]{(\ref{eq:#1})}

\newtheorem{theorem}{Theorem} %

\newtheorem{proposition}[theorem]{Proposition}

\newtheorem{assumption}{Assumption}

\DeclareMathOperator{\E}{\mathop{E}}

\DeclareMathOperator{\Cov}{Cov}
\newcommand{\trans}{^\top}
\newcommand{\best}{_o}
\newcommand{\besttheta}{{\theta\best}}
\newcommand{\AICc}{\text{AIC}_\text{C}} %
\newcommand{\XAICc}{\text{XAIC}_\text{C}}
\newcommand{\FAICc}{\text{FAIC}_\text{C}}

\newcommand{\RSS}{\mathrm{RSS}}
\DeclareMathOperator{\trace}{tr}

\newcommand{\R}{{\mathbf R}}

\begin{document}

\maketitle

\begin{abstract}
  Methods for combining predictions from different models in a
  supervised learning setting must somehow estimate/predict the
  quality of a model's predictions at unknown future inputs.
  Many of these methods (often implicitly) make the assumption that the
  test inputs are identical to the training inputs, which is seldom
  reasonable. By failing to take into account that prediction will
  generally be harder for test inputs that did not occur in the
  training set, this leads to the selection of too complex models.
  Based on a novel, unbiased expression for KL divergence, we propose
  XAIC and its special case FAIC as versions of AIC intended for
  prediction that use different degrees of knowledge of the test
  inputs. Both methods substantially differ from and may
  outperform all the known
  versions of AIC \emph{even when the training and test inputs are iid}, and
  are especially useful for deterministic inputs and under covariate
  shift.  Our experiments on linear models suggest that if the test
  and training inputs differ substantially, then XAIC and FAIC
  predictively outperform AIC, BIC and several other methods including
  Bayesian model averaging.

\end{abstract}

\section{INTRODUCTION}

In the statistical problem of model selection, we are given a set of
models $\{\, {\cal M}_i \mid i \in {\cal I} \,\}$, each of the form
${\cal M}_i = \{\, g_i(\cdot \mid \theta) \mid \theta \in
\Theta_i\,\}$, where the $g_i(\cdot \mid \theta)$ are density
functions on (sequences of) data. We wish to use one of these models
to explain our data and/or to make predictions of future data, but do
not know which model explains the data best. It is well known that
simply selecting the model containing the maximum likelihood
distribution from among all the models leads to overfitting, so any
expression of the quality of a model must somehow avoid this problem.
One way to do this is by estimating each model's ability to predict
\emph{unseen} data (this will be made precise below). This approach is
used by many methods for model selection,
including cross-validation, AIC \citep{Akaike} and its many variants,
\citeauthor{GelfandGhosh1998}'s $D_k$
\citeyearpar{GelfandGhosh1998}, and BPIC \citep{Ando2007}. However,
none of these methods takes into account that for supervised learning
problems, the generalization error being estimated will vary with the
test input variables. Instead, they implicitly assume that the test
inputs will be \emph{identical} to the training inputs.

In this paper, we derive an estimate of the generalization error that
does take the input data into account, and use this to define a new
model selection criterion XAIC, its special case FAIC, and the variants
$\XAICc$ and $\FAICc$ (small sample corrections). We use similar
assumptions as AIC, and thus our methods can be seen as relatives of
AIC that are adapted to supervised learning when the training and test
inputs differ. Our experiments show that our methods have excellent
predictive performance, better even than Bayesian model averaging in
some cases. Also, we show theoretically that AIC's unawareness of
input variables leads to a bias in the selected model order, even in
the seemingly safe case where the test inputs are drawn from the same
distribution as the training inputs. No existing model selection
method seems to address this issue adequately, making XAIC and FAIC
more than ``yet another version of AIC''. %

It is in fact quite surprising that, more than 40 years after its
original invention, all the forms of AIC currently in use are biased
in the above sense, and in theoretical analyses, conditional model
selection methods are often even compared on a new point $x$
constrained to be one of the $x$ values in the training data (see
e.g. \citet{Yang_shared}), even though in most practical
problems, a new point $x$ will \emph{not} be drawn from this empirical
training data distribution, but rather should be regarded as falling
in one of the three cases considered in this paper: (a) it is drawn
from the same distribution as the training data (but not necessarily
equal to one of the training inputs); (b) it is drawn from a
different distribution (covariate shift); (c) it is set to a fixed,
observable value, usually not in the training set, but the process
that gave rise to this value may not be known.

\subsection{GOALS OF MODEL SELECTION}\label{sec:goals}

When choosing among or combining predictions from different models,
one can have different goals in mind. Whereas BIC and BMS (Bayesian
model selection) focus on finding the most probable model, methods
like AIC, cross-validation %
and SRM
(structural risk minimization, \citet{Vapnik_SLT}) aim to find the
model that leads to the best {\em predictions\/} of future data.
While AIC and cross-validation typically lead to predictions that
converge faster to optimal in the sense of KL-divergence than those of
BIC and BMS, it is also well-known that, unlike BIC and BMS, such
methods are not statistically consistent (i.e. they do not find the
smallest submodel containing the truth with probability 1 as $n
\rightarrow \infty$); there is an inherent conflict between these two
goals, see for example \citet{Yang_simple,switchNIPS,switchJRSS}. Like AIC, the
XAIC and FAIC methods developed here aim for predictive optimality
rather than consistency, thus, if consistency is the main concern,
they should not be used.
We also stress at the outset that, unlike most other model selection
criteria, the model selected by FAIC may {\em depend\/} on the new $x$
whose corresponding $y$ value is to be predicted; for different $x$, a
different model may be selected based on the same training data. Since
--- as in many other model selection criteria --- our goal is
predictive accuracy rather than `finding the true model', and since
the dependence on the test $x$ helps us to get substantially better
predictions, we are not worried by this dependency.

FAIC thus cannot be said to select a `single' model for a given
training set --- it merely outputs a {\em function\/} from $x$ values
to models. As such, it is more comparable with BMA (Bayesian model
{\em averaging\/}) rather than BMS ({\em selection\/}). BMA is of
course a highly popular method for data prediction; like FAIC, it
adapts its predictions to the test input $x$ (as we will see, FAIC
tends to select a simpler model if there are not many training points
near $x$; BMA predicts with a larger variance if there are not many
training points near $x$). BMA leads to the optimal predictions in the
idealized setting where one takes expectation under the prior (i.e.,
in frequentist terms, we imagine nature to draw a model, and then a
distribution within the chosen model, both from the prior used in BMA,
and then data from the drawn distribution), and usually performs very
well in practice as well. %
It is of considerable interest then that our XAIC and FAIC
outperform Bayes by a fair margin in some of our
experiments in Section~\ref{sec:experiments}. %

\subsection{IN-SAMPLE AND EXTRA-SAMPLE ERROR} %

Many methods for model selection work by computing some estimate of
how well each model will do at predicting unseen data.
This generalization error may be defined in various ways, and methods
can further vary in the assumptions used to find an estimate. AIC
\citep{Akaike} is based on the expression for the generalization error
\begin{equation}\label{eq:AICtarget}
-2\E_{\mathbf U} \E_{\mathbf V}
\log g_i(\mathbf V \mid \hat\theta_i(\mathbf U)),
\end{equation}
for model ${\cal M}_i = \{\, g_i(\cdot \mid \theta) \mid \theta \in
\Theta_i\,\}$, 
where $\hat\theta_i(\mathbf U)$ denotes the element of $\Theta_i$
which maximizes the likelihood of data $\mathbf U$, and where both
random variables are independent samples of $n$ data points each, both
following the true distribution of the data. (We use capitals to
denote sequences of data points, and boldface for random variables.
Throughout this paper, $\log$ denotes the natural logarithm.)
Up to an additive term which is the same for all models, the inner
expectation is the KL divergence from the true distribution to
$g_i(\cdot \mid \hat\theta_i(\mathbf U))$. An interpretation of
\eqref{AICtarget} is that we first estimate the model's parameters
using a random sample $\mathbf U$, then judge the quality of this
estimate by looking at its performance on an independent, identically
distributed sample $\mathbf V$. AIC then works by estimating
\eqref{AICtarget} for each model %
by the asymptotically unbiased estimator
\begin{equation}\label{eq:AICearly}
  -2 \log g_i(\mathbf U \mid \hat\theta(\mathbf U)) + 2k,
\end{equation}
and selecting the model minimizing this estimate.
Thus AIC selects the model whose maximum likelihood
estimate is expected to be closest to the truth in terms of KL
divergence. In the sequel, we will consider only one model at a time,
and therefore omit the model index.

In supervised learning problems such as regression and classification,
the data points consist of two parts $u_i = (x_i, y_i)$, and the
models are sets of distributions on the \emph{output variable}
$\mathbf y$ conditional on the \emph{input variable} $x$ (which may or
may not be random). We call these \emph{conditional} models.
The
conditionality expresses that we are not interested in explaining the
behaviour of $x$, only that of $\mathbf y$ given $x$. Then \eqref{AICtarget}
can be adapted in two ways: as the \emph{extra-sample error}
\begin{equation}\label{eq:conditionalAICextrasample}
-2\E_{\mathbf Y \mid X} \E_{\mathbf Y' \mid X'}
\log g(\mathbf Y' \mid X', \hat\theta(X, \mathbf Y)),
\end{equation}
and, replacing both $X$ and $X'$ by a single variable $X$, as the
\emph{in-sample error}
\begin{equation}\label{eq:conditionalAICinsample}
-2\E_{\mathbf Y \mid X} \E_{\mathbf Y' \mid X}
\log g(\mathbf Y' \mid X, \hat\theta(X, \mathbf Y)),
\end{equation}
where capital letters again denote sequences of data points. %
Contrary
to \eqref{AICtarget}, these quantities capture that the expected
quality of a prediction regarding $\mathbf y$ may vary with $x$.

An example of a supervised learning setting is given by \emph{linear
  models}. In a linear model, an input variable $x$ is represented by
a \emph{design vector} and a sequence of $n$ inputs by an $n \times p$
\emph{design matrix}; with slight abuse of notation, we use $x$ and
$X$ to represent these. Then the densities $g(\mathbf Y \mid X, \mu)$ in the
model are Gaussian with mean $X \mu$ and covariance matrix $\sigma^2
I_n$ for some fixed $\sigma^2$. Because $g$ is of the form
$e^{-\text{squared error}}$, taking the negative logarithm as in
\eqref{AICtarget} produces an expression whose main component is a sum
of squared errors; the residual sum of squared errors $\RSS(\mathbf
Y)$ is the minimum for given data, which is attained by the maximum
likelihood estimator.  Alternatively, $\sigma^2$ may be another
parameter in addition to $\mu$ if the true variance is unknown.

It is standard to apply ordinary AIC to supervised learning
problems, %
for example for linear models with fixed variance where \eqref{AICearly} takes the well-known form
\begin{equation}\label{eq:AICRSS}
  \frac{1}{\sigma^2} \RSS(\mathbf Y) + 2k,
\end{equation}
where $k$ is the number of parameters in the model. But because the
standard expression behind AIC \eqref{AICtarget} makes no mention of
$X$ or $X'$, this corresponds to the tacit assumption that $X=X'$, so
that the in-sample error is being estimated.

However, the extra-sample error is more appropriate as a measure of
the expected performance on new data.
AIC was intended to correct the bias that results from evaluating an
estimator on the data from which it was derived, but because it uses
the in-sample error, AIC evaluates estimators on new output data, but
old input data. So we see that in supervised problems, a bias similar
to the one it was intended to correct is still present in AIC.

\subsection{CONTENTS}

The remainder of this article is structured as follows. In
Section~\ref{sec:main}, we develop our main results about the
extra-sample error and propose a new model selection criterion based
on this. It involves $\kappa_{X'}$, a term which can be calculated
explicitly for linear models; we concentrate on these models in the
remainder of the paper. Special cases of our criterion, including a
focused variant, are presented in Section~\ref{sec:nonfocused}. In
Section~\ref{sec:kappaBehaviour} we discuss the behaviour of our
estimate of the extra-sample error, and find that without our
modification, AIC's selected model orders are biased. Several
experiments on simulated data are described in
Section~\ref{sec:experiments}. Section~\ref{sec:discussion} contains
some further theoretical discussion regarding Bayesian prediction and
covariate shift. Finally, Section~\ref{sec:conclusion} concludes. All
proofs are in the \iftoggle{supplementary}{supplementary
  material.}{supplementary material.\footnote{Posted on arXiv}}

\section{ESTIMATING THE EXTRA-SAMPLE ERROR}\label{sec:main}

In this section, we will derive an estimate for the extra-sample
error. Our assumptions will be similar to those used in AIC to
estimate the in-sample error; therefore, we start with some
preliminaries about the setting of AIC.

\subsection{PRELIMINARIES}\label{sec:prelim}

In the setting of AIC, the data points are independent but not
necessarily identically distributed. The number of data points in
$\mathbf Y$ and $\mathbf Y'$ is $n$.
We define the Fisher information matrix $I(\theta)$ as
$-\E_{\mathbf Y'}\frac{\partial^2}{\partial\theta^2} \log g(\mathbf Y'
\mid \theta)$,
and define the conditional Fisher information matrix $I(\theta
\mid X')$ analogously. We write $\Cov(\hat\theta(X, \mathbf Y) \mid
X)$ for the conditional covariance matrix $\E_{\mathbf Y \mid X}
[\hat\theta(X, \mathbf Y) - \E_{\mathbf Y \mid X} \hat\theta(X, \mathbf Y)]
[\hat\theta(X, \mathbf Y) - \E_{\mathbf Y \mid X} \hat\theta(X, \mathbf Y)]\trans$.

Under standard regularity assumptions, there exists a unique parameter
value $\besttheta$ that minimizes the KL divergence from the true
distribution, and this is what $\hat\theta(\mathbf Y)$ converges to.
Under this and other (not very restrictive) regularity assumptions
\citep{Shibata89}, it can be shown that \citep{BurnhamAnderson}
\begin{equation}\label{eq:TIC}
  -2 \log g(\mathbf Y \mid \hat\theta(\mathbf Y))
  + 2\widehat{\trace}\left\{I(\besttheta) \Cov(\hat\theta(\mathbf Y))\right\}
\end{equation}
(where $\widehat{\trace}$ represents an appropriate estimator of that
trace) is an asymptotically unbiased estimator of \eqref{AICtarget}.
The model selection criterion TIC (Takeuchi's information criterion)
selects the model which minimizes \eqref{TIC}.

The estimator of the trace term that TIC
requires has a large variance, making it somewhat unreliable in
practice. AIC uses the very simple estimate $2k$ for TIC's trace
term. This estimate is generally biased except when the true
data-generating distribution is in the model, but obviously has $0$
variance. %
Also, if some models
are more misspecified than others, those models will have a worse
log-likelihood. This term in AIC grows linearly in the sample size, so
that asymptotically, those models will be disqualified by AIC. Thus
AIC selects good models even when its penalty term is biased due to
misspecification of the models.

This approach corresponds to making the following assumption in the
derivation leading to AIC's penalty term:
\begin{assumption}\label{ass:AIC}
  The model contains the true data-generating distribution.
\end{assumption}
It follows that $\theta\best$ specifies this distribution. We
emphasize that this assumption is only required for AIC's derivation
and does not mean that AIC necessarily works badly if applied to
misspecified models. Under this assumption, the two matrices in
\eqref{TIC} cancel, so the objective function becomes \eqref{AICearly}, the standard
formula for AIC \citep{BurnhamAnderson}.

We now move to supervised learning problems, where the true
distribution of the data and the distributions $g$ in the models are
conditional distributions of output values given input values. In this
setting, the data are essentially iid in the sense that $g(\mathbf Y
\mid X, \theta) = \prod_{i=1}^n g(\mathbf y_i \mid x_i, \theta)$.
That is, the outputs are independent given the inputs, and if two
input variables are equal, the corresponding output variables are
identically distributed. Also, the definition of $\besttheta$ would
need to be modified to depend on the training inputs, but since
Assumption~\ref{ass:AIC} now implies that $g(\mathbf y \mid x,
\besttheta)$ defines the true distribution of $\mathbf y$ given $x$
for all $x$, we can take this as the definition of $\besttheta$ for
supervised learning when Assumption~\ref{ass:AIC} holds.

For supervised learning problems, AIC and TIC silently assume that
$X'$ either equals $X$ or will be drawn from its empirical
distribution. We want to remove this assumption.

\subsection{MAIN RESULTS}

We will need another assumption:
\begin{assumption}\label{ass:FAIC}
  For training data $(X, \mathbf Y)$ and (unobserved) test data
  $(X', \mathbf Y')$,
  \begin{multline*}
    -\frac{1}{n}\E_{\mathbf Y \mid X} \log g(\mathbf Y \mid X, \theta\best) \\
  = -\frac{1}{n'}\E_{\mathbf Y' \mid X'} \log g(\mathbf Y' \mid X', \theta\best),
  \end{multline*}
  where $n$ and $n'$ denote the number of data points in $X$ and $X'$,
  respectively.
\end{assumption}
This assumption ensures that the log-likelihood on the test data can
be estimated from the training data. If $\mathbf X$ and $\mathbf X'$
are random and mutually iid, this is automatically satisfied when the
expectations are taken over these inputs as well. %
While this assumption of randomness is
standard in machine learning, there are other situations
where $X$ and $X'$ are not random and Assumption~\ref{ass:FAIC} holds
nevertheless.
For instance, this is the case if
$g(\mathbf y \mid x, \theta)$ is such that $\mathbf y_i =
f_\theta(x_i) + \mathbf z_i$, where the noise terms $\mathbf z_i$ are
zero-mean and iid (their distribution may depend on $\theta$). This
additive noise assumption is common in regression-like settings. Then
Assumption~\ref{ass:AIC} implies that Assumption~\ref{ass:FAIC} holds
for all $X, X'$. %

To get an estimator of the extra-sample error
\eqref{conditionalAICextrasample}, we do not make any assumptions
about the process generating $X$ and $X'$ but leave the variables
free. We allow $n \neq n'$.

\begin{theorem}\label{thm:FAIC}
  Under Assumptions~\ref{ass:AIC} and \ref{ass:FAIC} and some standard
  regularity conditions (detailed in the supplementary material), and
  for $n'$ either constant or growing with $n$,
  \begin{multline}\label{eq:FAICthm}
    -2\frac{n}{n'}\E_{\mathbf Y \mid X} \E_{\mathbf Y' \mid X'}
    \log g(\mathbf Y' \mid X', \hat\theta(X, \mathbf Y))\\
    = -2\E_{\mathbf Y \mid X} \log g(\mathbf Y \mid X,
    \hat\theta(X, \mathbf Y))
    + k + \kappa_{X'}
    + o(1),
  \end{multline}
  where
    $\kappa_{X'} = \frac{n}{n'}\trace\left\{I(\besttheta \mid X')
      \Cov(\hat\theta(X, \mathbf Y) \mid X)\right\}.$

  Moreover, if the true conditional distribution of
  $\mathbf Y$ given $X$ is Gaussian with fixed variance and the
  conditional distributions in the models are also Gaussian with that
  same variance (as is the case in linear models with known variance), then
  the above approximation becomes exact.
\end{theorem}

We wish to use \eqref{FAICthm} as a basis for model selection. To do
this, first note that \eqref{FAICthm} can be estimated from our training
data using
\begin{equation}\label{eq:FAICempirical}
  -2 \log g(\mathbf Y \mid X, \hat\theta(X, \mathbf Y)) + k + \kappa_{X'}.
\end{equation}
\emph{Theorem~\ref{thm:FAIC} expresses that this is an asymptotically
  unbiased estimator of the extra-sample error.} We see that the
difference with standard AIC \eqref{AICearly} is that the penalty $2k$
has been replaced by $k + \kappa_{X'}$. We propose to use
\eqref{FAICempirical} as the basis for a new model selection criterion
\emph{extra-sample AIC (XAIC)}, which chooses the model that minimizes
an estimator of \eqref{FAICempirical}. What remains for this is to
evaluate $\kappa_{X'}$, which may depend on the unknown true
distribution, and on the test set through $X'$.

\subsection{THE $\kappa_{X'}$ AND $o(1)$ TERMS FOR LINEAR MODELS}\label{sec:xaiclinmod}

If the densities $g$ are Gaussian, then $\kappa_{X'}$ does not depend
on the unknown $\theta\best$ because the Fisher information is
constant, so no additional estimation is necessary to evaluate
it. Thus for a linear model with fixed variance, $\kappa_{X'}$ becomes
\begin{multline*}
\kappa_{X'}
= \frac{n}{n'}\trace\left\{\left[\frac{1}{\sigma^2} {X'}\trans X' \right]
  \left[\sigma^2(X\trans X)^{-1}\right]\right\}\\
= \frac{n}{n'}\trace\left[ {X'}\trans X' (X\trans X)^{-1} \right].
\end{multline*}
If the variance is also to be estimated, it can be easily seen that
$\kappa_{X'}$ will become this value plus one. In that case, the
approximation in Theorem~\ref{thm:FAIC} is not exact (as it is in the
known variance case), but the $o(1)$ term can be evaluated explicitly:

\begin{theorem}\label{thm:FAICc}
  For a linear model with unknown variance,
  \begin{multline*}
    -2\frac{n}{n'}\E_{\mathbf Y \mid X} \E_{\mathbf Y' \mid X'}
    \log g(\mathbf Y' \mid X', \hat\theta(X, \mathbf Y))\\
    = -2\E_{\mathbf Y \mid X} \log g(\mathbf Y \mid X, \hat\theta(X,
    \mathbf Y))\\
    + k + \kappa_{X'} +
    \frac{(k+\kappa_{X'})(k+1)}{n-k-1},
  \end{multline*}
  where $\kappa_{X'}$ can again be computed from the data and equals
  $(n/n')\trace({X'}\trans X'(X\trans X)^{-1}) + 1$, and $k$ is the
  number of parameters including $\sigma^2$.
\end{theorem}

Theorem~\ref{thm:FAICc} presents an extra-sample analogue of the
well-known small sample correction $\AICc$ \citep{HurvTsai}, which is
derived similarly and uses a penalty of $2k +
2k(k+1)/(n-k-1)$. We define $\XAICc$ accordingly.
Though the theorem holds exactly only in the specific case described,
we believe 
that the extra penalty term %
will lead to better results in much more general settings in practice,
as is the case with $\AICc$ \citep{BurnhamAnderson}.

\section{MODEL SELECTION FOR EXTRA-SAMPLE PREDICTION}\label{sec:nonfocused} %

In this section, we discuss several concrete model selection methods,
all based on the XAIC formula \eqref{FAICempirical} and thus
correcting AIC's bias.

\subsection{NONFOCUSED VERSIONS OF XAIC}

Except in trivial cases, the extra-sample error
\eqref{conditionalAICextrasample} and its estimate
\eqref{FAICempirical} depend on the test inputs $X'$, so some
knowledge of $X'$ is required when choosing a model appropriate for
extra-sample prediction. In a semi-supervised learning setting where
$X'$ itself is known at the time of model selection, we could evaluate
\eqref{FAICempirical} directly for each model. However, $X'$ might not
yet be known when choosing a model.

If $\mathbf X'$ is not known but its distribution is,
we can replace $\kappa_{\mathbf X'}$ by its expectation;
for iid inputs, computing this reduces to computing $\E_{\mathbf x'}
I(\besttheta \mid \mathbf x')$. %

If the distribution of $\mathbf X'$ is also unknown, we need to
estimate it somehow.
If it is believed that $\mathbf X$ and $\mathbf X'$ follow the same
distribution, the empirical distribution of $\mathbf X$ could be used
as an estimate of the distribution of $\mathbf X'$. Then 
AIC is retrieved as a special case. Section~\ref{sec:kappaBehaviour} will show that
this is a bad choice even if $\mathbf X$ and $\mathbf X'$ follow the
same distribution, so a smoothed estimate is recommended instead.

Of course, we are not restricted to the case where $\mathbf X$ and
$\mathbf X'$ follow similar distributions. In the setting of covariate
shift \citep{SugiyamaCovShiftBook}, the distributions are different
but known (or can be estimated). This variant of XAIC is directly
applicable to that setting, yielding an unbiased analogue of AIC.

\subsection{FOCUSED MODEL SELECTION}

It turns out there is a way to apply \eqref{FAICempirical} even when
nothing is known about the process generating $X$ and $X'$. If our
goal is prediction, we can set $X'$ to the single point $x'$ for
which we need to predict the corresponding $\mathbf y'$. Contrary to
standard model selection approaches, we thus use $x'$ already at the
stage of model selection, rather than only inside the models. We
define the model selection criterion \emph{Focused AIC (FAIC)} as this
special case of XAIC, and $\FAICc$ as its small sample correction.

A focused model selection method implements the intuition that those
test points whose input is further away from the training inputs
should be predicted with more caution; that is, with less complex
models. As discussed in Section~\ref{sec:goals}, methods that optimize
predictive performance often are not consistent; this hurts in
particular for test inputs far away from the training inputs. We
expect that extra-sample adaptations of such methods (like XAIC) are
also inconsistent, but that using the focused special case helps to
guard against
this small chance of large loss.

Choosing a model specifically for the task at hand potentially lets us
end up with a model that performs this task much better than a model
found by a non-focused model selection method. However, there are
situations in which focus is not a desirable property: the mapping
from input values to predictions given by a focused model selection
method will be harder to interpret than that of a non-focused method,
as it is a combination of the models under consideration rather than a
single one of them. Thus, if the experimenter's goal is
interpretation/transparency, %
a focused model selection method is not recommended; these methods are
best applied when the goal is prediction. 

Evaluating the $x'$-dependent model selection criterion separately for
each $x'$ leads to a regression curve which in general will not be from
any one of the candidate models, but only piecewise so. It will
usually have discontinuities where it switches between models. If the
models contain only continuous functions and such discontinuities are
undesirable, Akaike weights \citep{Akaike1979, BurnhamAnderson} may be
used to get a continuous analogue of the FAIC regression curve.

\section{AIC VS XAIC ($k$ VS $\kappa_{x}$) IN LINEAR MODELS}\label{sec:kappaBehaviour}

Intuitively, the quantity $\kappa_x$ that appears as a penalty term in
the XAIC formula \eqref{FAICempirical} expresses a measure of
dissimilarity between the test input $x$ and the training inputs $X$.
This measure is determined fully by the models and does not have to be
chosen some other way. However, its properties are not readily
apparent. Because $\kappa_x$ can be computed exactly for linear
models, we investigate some of its properties in that case.

One useful characterization of $\kappa_x$ is the following: if we
express the design vector $x$ of the test point in a basis that is
orthonormal to the empirical measure of the training set $X$, then
$\kappa_x = \lVert x \rVert^2$.

For given $X$, $x$ may exist such that $\kappa_x$ is either greater or
smaller than the number of parameters $k$. An example of $\kappa_x <
k$ occurs for the
linear model consisting of all linear functions with known variance
(so $k=2$). Then $\kappa_x$ will be minimized when $x$ lies at the
mean of the input values in the training set, where $\kappa_x = 1$.

We will now consider the case where $\mathbf X$ and $\mathbf x$ are
random and iid. We showed that the XAIC expression
\eqref{FAICempirical} is an unbiased estimator of the extra-sample
error. AIC uses $k$ in place of $\kappa_{\mathbf x}$, and the above
suggests the possibility that maybe the instances where
$\kappa_{\mathbf x} > k$ and those where $\kappa_{\mathbf x} < k$
cancel each other out, so that AIC would also be approximately
unbiased as an estimate of the extra-sample error. However, the
following proposition shows that, except in a trivial case,
$\kappa_{\mathbf x}$ is on average greater than $k$. This means that
in those settings, AIC underestimates the model's extra-sample error.

(We should mention here that if $\mathbf X$ and $\mathbf x$ are random
and mutually iid, then as $n \to \infty$, AIC's bias goes to 0. The
bias we show below concerns all finite $n$; additionally, without
focus, an extreme $\mathbf x$ could result in a very biased AIC value
even for large $n$.)

\begin{proposition}\label{prop:bias}
  Consider a linear model ${\cal M}$ with training inputs $\mathbf X$
  and test input $\mathbf x$ iid such that $\mathbf X\trans \mathbf X$
  is almost surely invertible. Let ${\cal M}'$ be the submodel
  obtained by removing the final entry from every design vector. Then
  these models are related by $\E\kappa_{\mathbf x} \geq
  \E\kappa_{\mathbf x'} + 1$, with strict inequality if $\mathbf x$
  has at least two entries.
\end{proposition}

It follows by induction on $k$ that for random input data, AIC is
biased as an estimate of the extra-sample error except in a special
case with $k=1$. Also, the bias becomes worse for larger models. This
last fact is distressing, as it shows that when AIC assesses a
sequence of nested models, the amount by which it overestimates their
generalization ability grows with the model order. Thus the biases in
the AIC scores lead to a bias in the selected model order, which was
not evident from earlier work.

The XAIC formula \eqref{FAICempirical} contains two terms
that depend on the data: minus two times the log-likelihood, and the
penalty term $\kappa_{X'}$. The log-likelihood measures distances
between output values and is independent of $X'$, while $\kappa_{X'}$
expresses a property of input values and is largely unaffected by
output values; in fact, in linear models its computation does not
involve any (estimates based on) output values. Hence the variance of
XAIC is no greater than that of AIC when comparing the two on fixed $X, X'$, so that
XAIC's reduction in bias does not come at the price of an increase in
variance. However, focused model selection demands that $X'$ is
\emph{not} held fixed, so that FAIC may have a larger variance than AIC.
Similarly, if the distribution of $\mathbf X'$ is being estimated as
part of applying XAIC, the used estimator's quality will affect the
accuracy of the estimated generalization error.

\section{EXPERIMENTS}\label{sec:experiments}

We will now experimentally compare XAIC and FAIC (or more precisely,
their small-sample corrected versions $\XAICc$ and $\FAICc$) to several other model
selection methods, in univariate and multivariate problems. 

\subsection{DESCRIPTION OF EXPERIMENTS}

In the univariate experiments, linear models
${\cal M}_1, \ldots, {\cal M}_7$ with unknown variance were
considered. Model ${\cal M}_i$ contained polynomials of degree $i-1$
(and so had $i+1$ parameters). The input values $x$ of the training
data were drawn from a Gaussian distribution with mean 0 and variance
1, while the output values were generated as $\mathbf y_i = f(x_i) +
\mathbf z$ with $\mathbf z_i$ iid Gaussians with mean 0 and variance
$0.1$, and $f$ some unknown true function. Given 100 training data
points, each of the eight model selection methods under consideration
had to select a model. The squared risk $(\hat{y} -
f(x))^2$ of the chosen model's prediction $\hat{y}$ was computed for
each of a range of values of the test point's $x$,
averaged over 100 draws of the training data. This experiment was
performed for two different true functions: $f_1(x) = x + 2$ and
$f_2(x) = \lvert x \rvert$.

In the multivariate experiments, each input variable was a vector
$(u_1, \ldots, u_6)$, and the models corresponded to all possible
subsets of these 6 variables. Each model also included an intercept
and a variance parameter. The true function was given by $f(u) = 2 +
u_1 + 0.1u_2 + 0.03u_3 + 0.001u_4 + 0.003u_5$, and the additive noise
was again Gaussian with variance $0.1$. A set of $n' = 400$ test
inputs was drawn from a standard Gaussian distribution, but the
training inputs were generated differently in each experiment: from
the same Gaussian distribution as the test inputs; from a uniform
distribution on $[-\sqrt{3}, \sqrt{3}]^6$; or from a uniform
`spike-and-slab' mixture of two Gaussians with covariance matrices
$(1/5)I_6$ and $(9/5)I_6$. Note that all three distributions have the
same mean and covariance as the test input distribution, making these
mild cases of covariate shift. For the Gaussian training case, we
report the results for $n=60$ and, after extending the same training
set, for $n=100$. Squared risks were averaged over
the test set and further over 50 repeats of these experiments.

The experiments used the version of XAIC that is given a distribution
of the test inputs, but not the test inputs themselves. In the
multivariate experiments, XAIC used the actual (Gaussian) distribution
of the test inputs. In the univariate case, two instances of XAIC were
evaluated: one for test inputs drawn from the same
distribution as the training inputs (standard Gaussian), and another
(labelled $\XAICc$2) for a Gaussian test input distribution with mean 0 and
variance 4.

Bayesian model averaging (BMA) differs from the other methods in that
it does not select a single model, but formulates its prediction as a
weighted average over them; in our case, its prediction corresponds to
the posterior mean over all models. Weighted versions exist of other
model selection methods as well, such as Akaike weights \citep{Akaike1979, BurnhamAnderson} for AIC and variants. In our experiments we saw that
these usually perform similar to but somewhat better than their
originals. In our univariate experiments, we decided against reporting
these, as they are less standard. However, in the multivariate
experiments, the weighted versions were all better than their
selection counterparts, so both are reported separately to allow fair
comparisons.

In our experiments, BMA used a uniform prior over the models. Within
the models, Jeffreys' noninformative prior (for which the selected
$\mu$ would correspond to the maximum likelihood $\hat\mu$ used by
other methods) was used for the variable selection experiments; for
the polynomial case, it proved too numerically unstable for the larger
models, so there BMA uses a weakly informative Gaussian prior
(variance $10^2$ on $\mu_2, \ldots, \mu_7$ with respect to the
Hermite polynomial basis, and Jeffreys' prior on $\sigma^2$).

Of the model selection methods included in our experiments, AIC was
extensively discussed in Section~\ref{sec:prelim}; as with XAIC and
FAIC, we use here the small sample correction $\AICc$ (see
Section~\ref{sec:xaiclinmod}). BIC \citep{SchwarzBIC} and BMS were
mentioned in Section~\ref{sec:goals} as methods that attempt to find
the most probable model given the data rather than aiming to optimize
predictive performance; both are based on BMA, which computes the
Bayesian posterior probability of each model. Three other methods were
evaluated in our experiments; these are discussed below.

Like AIC, the much more recent focused information criterion (FIC)
\citep{FIC} is designed to %
make good predictions. Unlike other methods,
these predictions are for a \emph{focus parameter} which may be any
function of the model's estimate, %
not just its prediction at some input value (though we only used the
latter in our experiments). Unlike FAIC, it uses this focus not just
for estimating a model's variance, but also its bias; FAIC on the
other hand uses a global estimate of a model's bias based on
Assumption~\ref{ass:FAIC}. A model's bias for the focus parameter is
evaluated by comparing its estimate to that of the most complex model
available.

Another more recent method for model selection is the
subspace information criterion (SIC) \citep{SIC}, which is applicable
to supervised learning problems when our models are subspaces of some
Hilbert space of functions, and our objective is to minimize the
squared norm. Like FIC, SIC estimates the models' biases by comparing
their estimates to that of a larger model, but it includes a term to
correct for this large model's variance. In our experiments, we used
the corrected SIC (cSIC) which truncates the bias estimate at 0.

Generalized cross-validation (GCV) \citep{GCV} can be seen as a
computationally efficient approximation of leave-one-out
cross-validation for linear models. We included it in our experiments
because \citet{LeebOutOfSample} shows that it performs better than
other model selection methods when the test input variables are newly
sampled.

\subsection{RESULTS}

\begin{figure}[t] %
\centering
\includegraphics[width=3.25in]{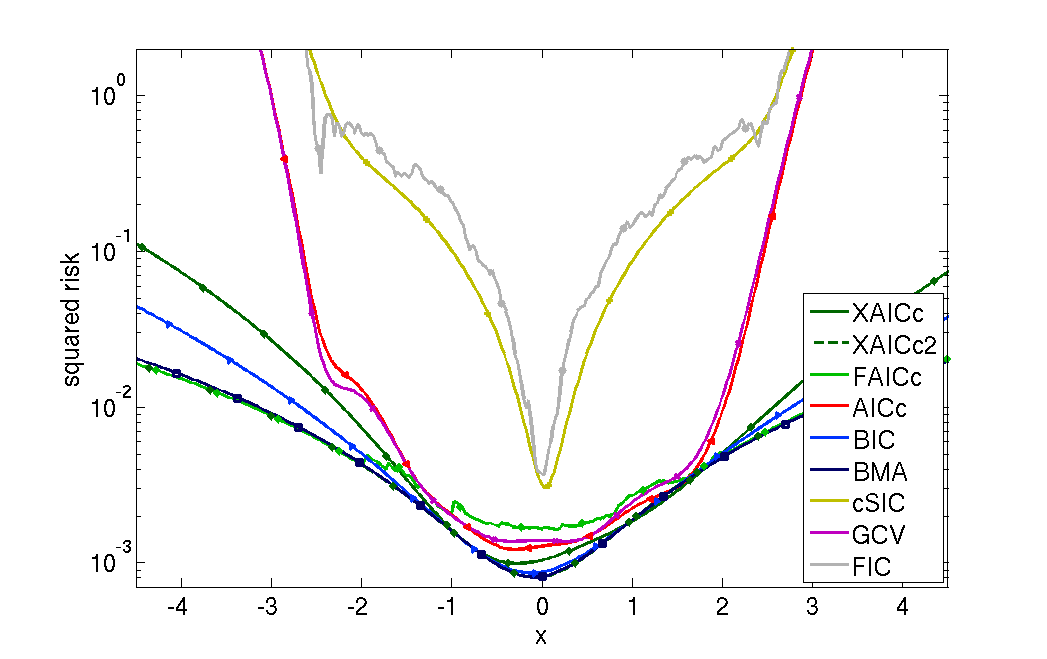}
\caption{Squared risk of different model selection methods
  as a function of $x$ when the true function is $f_1(x) = x +
  2$.}\label{fig:risk_line}
\end{figure}
\begin{figure}[t]
\centering
\includegraphics[width=3.25in]{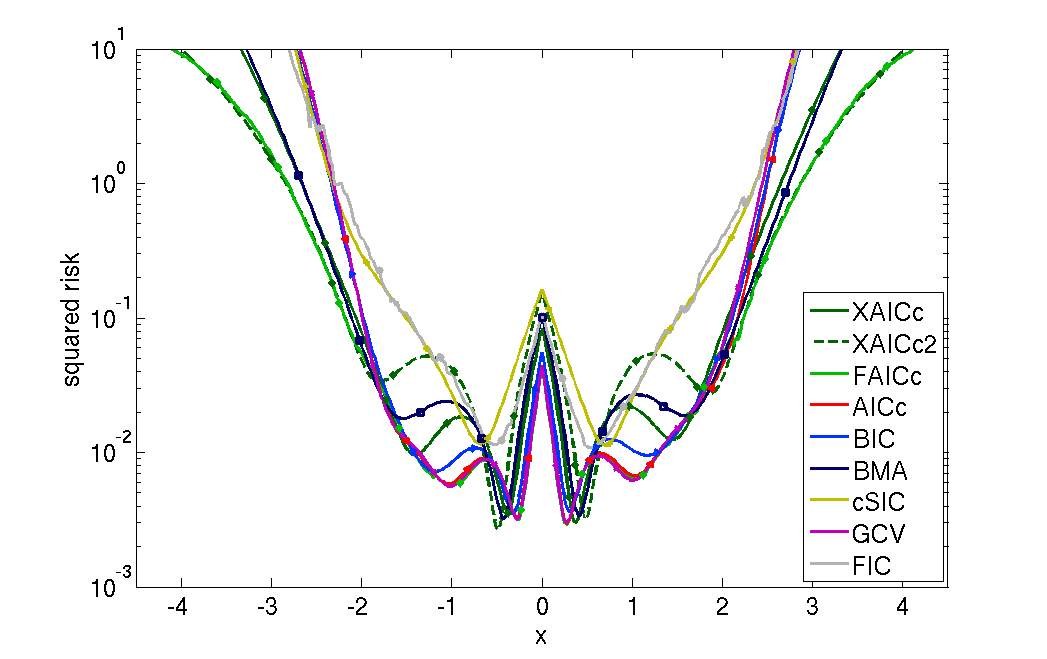}
\caption{Squared risk of different model selection methods as
  a function of $x$ when the true function is $f_2(x) = \lvert x
  \rvert$.}\label{fig:risk_abs}
\end{figure}
\begin{table}[t]
\caption{Average selected model index per method for $f_1$ and $f_2$,
  at test inputs $x'=0$ and $4$ (if different).}\label{tab:selections}
\begin{center}
\setlength{\tabcolsep}{3pt}
\begin{tabular}{r|ccccccc|}
&$\XAICc$&$\XAICc$2&$\AICc$&
 BIC &   BMS &  cSIC &   GCV \\
\hline
$f_1$ &    2.10 &  2.00 & 2.33 & 2.02 & 2.00 & 2.94 & 2.38 \\ %
$f_2$ &    4.57 &  3.00 & 6.38 & 5.70 & 4.05 & 4.70 & 6.49 \\ %
\hline
\end{tabular}

\vspace*{2mm}

\begin{tabular}{r|cc|cc|}
& \multicolumn{2}{c}{$\FAICc$}& \multicolumn{2}{|c|}{FIC} \\
& $x'=0$ & $x'=4$ & $x'=0$ & $x'=4$ \\
\hline
$f_1$&    2.94 &  2.00 & 2.66 & 3.12 \\ %
$f_2$&    6.56 &  1.54 &  5.29 &  5.35 \\ %
\hline
\end{tabular}
\end{center}
\end{table}
\begin{table}[t]
\caption{Multivariate: squared risk for different training
  sets; model selection}\label{tab:multivarselection}
\begin{center}
\setlength{\tabcolsep}{3pt}
\begin{tabular}{l|cccc|}
        &          &         & spike-     &       \\
        & Gaussian & uniform & and-slab & Gaussian \\
        & ($n=60$) & ($n=60$) & ($n=60$) & ($n=100$) \\
\hline
$\XAICc$ & 0.0119 & 0.0123 & 0.0144 & 0.0070 \\
$\FAICc$ & 0.0123 & 0.0127 & 0.0133 & 0.0077 \\
$\AICc$  & 0.0125 & 0.0126 & 0.0156 & 0.0070 \\
BIC      & 0.0113 & 0.0128 & 0.0140 & 0.0073 \\
BMS      & 0.0120 & 0.0126 & 0.0138 & 0.0075 \\
cSIC     & 0.0119 & 0.0134 & 0.0138 & 0.0074 \\
GCV      & 0.0129 & 0.0131 & 0.0153 & 0.0072 \\
FIC      & 0.0196 & 0.0189 & 0.0241 & 0.0111 \\
\hline
\end{tabular}
\end{center}
\end{table}
\begin{table}[t]
\caption{Multivariate: squared risk for different training
  sets; model weighting/averaging}\label{tab:multivarweighting}
\begin{center}
\setlength{\tabcolsep}{3pt}
\begin{tabular}{l|cccc|}
        &          &         & spike-     &       \\
        & Gaussian & uniform & and-slab & Gaussian \\
        & ($n=60$) & ($n=60$) & ($n=60$) & ($n=100$) \\
\hline
$\XAICc$w & 0.0099 & 0.0108 & 0.0114 & 0.0063 \\
$\FAICc$w & 0.0100 & 0.0110 & 0.0110 & 0.0066 \\
$\AICc$w  & 0.0101 & 0.0108 & 0.0119 & 0.0063 \\
BICw      & 0.0096 & 0.0106 & 0.0111 & 0.0062 \\
BMA      & 0.0100 & 0.0107 & 0.0113 & 0.0061 \\
\hline
\end{tabular}
\end{center}
\end{table}
Results from the two univariate experiments are shown in
Figures~\ref{fig:risk_line} and \ref{fig:risk_abs} (squared risks) and
in Table~\ref{tab:selections} (selected models). Squared risk results
for the multivariate experiments are given in
Table~\ref{tab:multivarselection} for the model selection methods, and
in Table~\ref{tab:multivarweighting} for the model weighting/averaging
variants.

\paragraph{XAIC and FAIC} The characteristic behaviours
of our methods are clearly visible in the univariate experiments. Both
instances of XAIC perform well overall in both experiment. Of the two,
$\XAICc$2 was set up to expect test inputs further away from the
center. As a result, it selects models more conservatively, and
obtains smaller risk at such off-center test inputs. Its selections
were very stable: in both experiments, $\XAICc$2 selected the same
model in each of the 100 runs.

We see that in the center of Figure~\ref{fig:risk_abs}, the simple
model chosen by $\XAICc$2 was outperformed by more complex
models. FAIC exploits this by choosing a model adaptively for each
test input. This resulted in good risk performance at all test inputs.

In the multivariate experiments, FAIC was the best method for the
spike-and-slab training data, where there are pronounced differences
in training point density surrounding different test points, so that
selecting a different model for each pays off. The performance of
XAIC was more reliable overall, comparing very favourably to each of
its competitors.

\paragraph{AIC} %
Our methods XAIC and FAIC were derived as adaptations of AIC, and share its
tendency to go for complex models as soon as there is some indication
that their predictions might be worthwhile. This leads to good
predictions on average, but also to inconsistency: when a simpler
model contains the true distribution, AIC will continue to select more
complex models with positive probability, no matter how large $n$
grows. This may sometimes hurt predictive performance, because the
accuracy of the estimated parameter will be smaller for more complex
models; for details, we refer to \citet{Yang_simple,switchNIPS,switchJRSS}.
XAIC makes a better assessment of the generalization error, even when
the training and test inputs follow the same distribution, so that it
overfits less than AIC and may achieve much better risks.
FAIC differs from AIC in another way:
its tendency to choose more complex models is strengthened in areas where
many data points are available (so that the potential damage of
picking an overly complex model is smaller), while it is suppressed
when few data points are available (and the potential damage is much
greater).

This tendency is also apparent in Table~\ref{tab:selections}. In the
first experiment, where a small model contains the true distribution,
it causes FAIC to perform worse than AIC near $x=0$. However, note
that the vertical axis is logarithmic, so the difference appears
larger than it is: when we average over the training input
distribution, we find that FAIC performs better by a factor 20 in
terms of squared risk.

In the multivariate experiments, XAIC again performs better than AIC,
though the difference eventually disappears as $n$ grows. With the
notable exception of the spike-and-slab experiment, FAIC does not
perform well here: in two of the experiments, it does worse than
AIC. Part of the reason must be our observation at the end of
Section~\ref{sec:kappaBehaviour}: FAIC's estimate of the
generalization error, while unbiased, may potentially have a larger
variance than (X)AIC's estimate, and this is not always a good trade-off.

\paragraph{BIC and BMS/BMA} BIC and BMS do not try to identify the
model that will give the best predictions now, but instead attempt to
find the most probable model given the data, which usually amounts to
the simplest model containing the true distribution.
This leads them to be conservative about selecting complex models. For
similar reasons, Bayesian model averaging (BMA) puts only small weight
on complex models.  We see this in Figure~\ref{fig:risk_line}, where
BIC and BMA have good performance because they most often select the
optimal second model (or in the case of BMA, give it the largest
weight). However, for $f_2$ in Figure~\ref{fig:risk_abs}, it may be
outperformed by FAIC or XAIC for test inputs away from the center. In
the multivariate experiments, XAIC often performs better than BMS/BMA,
and rarely much worse; the only instance of the latter is for the
spike-and-slab data, where FAIC outperforms both. (See
Section~\ref{sec:bayespred} for further discussion of BMA.)

\paragraph{FIC} In all our experiments, FIC obtained large squared risks, and we see in Table~\ref{tab:selections} that its
selection behaviour was the opposite of FAIC: for extreme $x$, FIC
often selects a more complex model than near $x=0$.  This seems to
happen because FIC uses the most complex model's prediction at a given
$x$ to estimate each other model's bias.  Because the most complex
model will usually have a significant variance, this resulted in FIC
being misled in many of the experiments we examined. In particular, in
areas with few training inputs, FIC apparently usually believes the simpler
models will perform badly because it attributes to them a large bias,
so that the same model as elsewhere (or even a more complex one) is
selected. Conversely, FIC was often observed to switch to an overly
simple model near some input value where this model's estimate
happened to coincide with that of the most complex model.

\paragraph{SIC} SIC obtained large risks in the univariate experiments due to underfitting. Its results in three of the four multivariate experiments were competitive, however.
\paragraph{GCV} Based on \citet{LeebOutOfSample}, we expected GCV might be one of the strongest competitors to XAIC. This was not
clearly reflected in our experiments, where its performance was very
similar to that of AIC.

\section{DISCUSSION}\label{sec:discussion}

\subsection{RELATION TO THE BAYESIAN PREDICTIVE DISTRIBUTION}\label{sec:bayespred}

The quantity $\kappa_{x'}$ that occurs in FAIC has an interpretation
in the Bayesian framework. If we do linear regression with known
variance and a noninformative prior on $\mu$, then after observing
$X$, $\mathbf Y$ and $x'$, the predictive distribution of $\mathbf y'$ is
$\mathbf y' \mid \mathbf Y, X, x' \sim {\cal N}({x'}\trans \hat\mu,
\sigma^2(1 + {x'}\trans (X\trans X)^{-1} x'))$. We see that
$\kappa_{x'}$ and the variance of this predictive distribution obey a
linear relation.
Thus if BMA is allowed to give a distribution over output values as
its prediction, then this distribution (a mixture of Gaussians with
different variances) will reflect that some models'
predictions are more reliable than others. %
However, if the predictive distribution must be summarized by a point
prediction, then such information is likely to be lost. For instance,
if the point prediction $\hat{y}'$ is to be evaluated with squared
loss and $\hat{y}'$ is chosen to minimize the expected loss under the
predictive distribution (as in our experiments in Section~\ref{sec:experiments}), then $\hat{y}'$ is a
weighted average of posterior means for $\mathbf y'$ given $x'$ (one
mean for each model, weighted by its posterior probability). The
predictive variances are not factored into $\hat{y}'$, so that in
this scenario, BMA does not use the information captured by $\kappa_{X'}$
that XAIC and FAIC rely on.

This is not to say that BMA \emph{should} use this information: the
consideration of finding the most probable model (BMS, BIC) or the
full distribution over models (BMA) is not affected by the purpose for
which the model will be used, so it should not depend on the input
values in the test data through $\kappa_{X'}$. This suggests that
there is no XBIC analogue to XAIC. For Bayesian methods such as DIC
\citep{DIC} and BPIC \citep{Ando2007} that aim for good predictions,
on the other hand, extra-sample and focused equivalents may exist.

\subsection{RELATION TO COVARIATE SHIFT} %

We observed at the end of Section~\ref{sec:kappaBehaviour} that of the two
data-dependent terms in XAIC, the log-likelihood is independent of
$X'$, while $\kappa_{X'}$ is (largely) unaffected by output values. An
important practical consequence of this split between input and output
values is that XAIC and FAIC look for models that give a good overall
fit, not just a good fit at the test inputs. $X'$ is then used to
determine how well we can expect these models to generalize to the
test set.
So if we have two models and
believe each to be able to give a good fit in a different region of
the input space, then FAIC is not the proper tool for the task of
finding these regions: FAIC considers global fit rather than local fit
when evaluating a model, and within the model selects the maximum
likelihood estimator, not an estimator specifically chosen for a local
fit at input point $x$.

In this respect, our methods differ from those commonly used in the
covariate shift literature (see \citet{SugiyamaCovShiftBook,TLSurvey};
some negative results are in \citet{DomAdaptImposs}), where typically
a model (and an estimator within that model) is sought that will
perform well on the test set only, using for example importance
weighting. This is appropriate if we believe that no available model
can give satisfactory results on both training and test inputs
simultaneously. In situations where such models are believed to exist,
our methods try to find them using all information in the training
set.

\section{CONCLUSIONS AND FUTURE WORK}\label{sec:conclusion}

We have shown a bias in AIC when it is applied to supervised learning
problems, and proposed XAIC and FAIC as versions of AIC which correct
this bias. We have experimentally shown that these methods give better
predictive performance than other
methods in many situations.%

We see several directions for future work. First, the practical
usefulness of our methods needs to be confirmed by further experiments.
Other future work includes considering other model selection methods:
determining whether they are affected by the same bias that we found
for AIC, whether such a bias can be removed (possibly leading to
extra-sample and focused versions of those methods), and how these
methods perform in simulation experiments and on real data.
In particular, BPIC \citep{Ando2007} is a promising candidate, as its
derivation starts with a Bayesian equivalent of \eqref{AICtarget}. An
XBPIC method would also be better able to deal with more complex models that
a variant of AIC would have difficulty with, such as hierarchical
Bayesian models, greatly increasing its practical applicability.

\subsubsection*{Acknowledgements}

I thank Peter Gr\"{u}nwald and Steven de Rooij for their valuable
comments and encouragement.

\newpage

\bibliographystyle{plainnat}
\DeclareRobustCommand{\VAN}[3]{#3}
\bibliography{bib}{}

\iftoggle{supplementary}{
\newpage

\section*{SUPPLEMENTARY MATERIAL}

\begin{assumption}[Regularity conditions]\label{ass:TIC}
  Items 1--4 correspond to the regularity assumptions behind AIC given
  by \citet{Shibata89}, but rewritten to take the input variables into
  account. Item 5 is the assumption of asymptotic normality of the
  maximum likelihood estimator, which is also standard.
  \begin{enumerate}
  \item $\Theta \subseteq \R^k$ is open,
    and for sufficiently large $n$ the gradient and Hessian of the
    log-likelihood function $\ell(\theta) = \log g(\mathbf Y \mid X, \theta)$
    are well-defined for all $\theta \in \Theta$ with probability 1,
    and both are continuous;

  \item For sufficiently large $n$, $\E_{\mathbf Y \mid X}
    |\frac{\partial}{\partial\theta}\ell(\theta)| < \infty$ and
    $\E_{\mathbf Y \mid X}
    |\frac{\partial^2}{\partial\theta^2}\ell(\theta)| < \infty$;

  \item For sufficiently large $n$, there exists a unique $\theta\best
    \in \Theta$ such that $\E_{\mathbf Y \mid X}
    \frac{\partial}{\partial\theta}\ell(\theta\best) = 0$. For all
    $\epsilon > 0$, it satisfies
  \begin{equation*}
    \inf_{\theta: \lVert\theta - \theta\best\rVert>\epsilon} \ell(\theta\best) -
    \ell(\theta) \to \infty \qquad\text{almost surely}
  \end{equation*}
  as $n \to \infty$;

  \item For all $\epsilon > 0$, there is a $\delta > 0$ such that for
    sufficiently large $n$,
    \begin{multline*}
      \scriptstyle\sup_{\lVert \theta - \besttheta \rVert < \delta} \big\lvert 
       \E_{\mathbf Y \mid X} [\hat\theta(\mathbf Y \mid X) - \besttheta]\trans
      I(\theta \mid X)
[\hat\theta(\mathbf Y \mid X) - \besttheta] \\
      \scriptstyle - \trace\left[J(\besttheta \mid X) I(\besttheta \mid X)^{-1} \right] \big\rvert
      < \epsilon, %
    \end{multline*}
    where $I(\theta \mid X) = -\E_{\mathbf Y \mid X}\frac{\partial^2}{\partial\theta^2} \ell(\theta)$ and $J(\besttheta \mid X) = \E_{\mathbf Y \mid X} \left[
      \frac{\partial}{\partial\theta}\ell(\theta\best) \right] \left[
      \frac{\partial}{\partial\theta}\ell(\theta\best) \right]\trans$
    are continuous and positive definite.

  \item $\sqrt{n} (\hat\theta(\mathbf Y \mid X) - \besttheta)
    \xrightarrow{D} {\cal N}(0, \Sigma)$ for some $\Sigma$.
  \end{enumerate}
\end{assumption}

\begin{proof}[Proof of Theorem \ref{thm:FAIC}]
  This proof is adapted from the one in \citet{BurnhamAnderson}, with
  modifications to take $X$ and $X'$ into account. Derivation of an
  estimator for \eqref{conditionalAICextrasample} starts with a Taylor
  expansion:
\begin{multline*}
-2 \log g(\mathbf Y' \mid X', \hat\theta(X, \mathbf Y))
= -2 \log g(\mathbf Y' \mid X', \besttheta)\\
-2\left[\frac{\partial}{\partial\theta} \log g(\mathbf Y' \mid X',
  \besttheta)\right]\trans
  [\hat\theta(X, \mathbf Y) - \besttheta]\\
-[\hat\theta(X, \mathbf Y) - \besttheta]\trans
  \left[\frac{\partial^2}{\partial\theta^2} \log g(\mathbf Y' \mid X',
  \besttheta)\right]
  [\hat\theta(X, \mathbf Y) - \besttheta]\\
+r(\hat\theta),
\end{multline*}
where $r(\hat\theta)/\lVert \hat\theta(X, \mathbf Y) - \besttheta
\rVert^2 \to 0$ as $\hat\theta(X, \mathbf Y) \to \besttheta$. We take
the expectation $\E_{\mathbf Y' \mid X'}$; given the regularity
conditions on the model, $\besttheta$ minimizes $\E_{\mathbf Y' \mid
  X'} \log g(\mathbf Y' \mid X', \theta)$, so the linear term
vanishes. (Note that we need this vanishing to hold for any $X'$ (or
equivalently, for any single point $x$); this follows from the
assumption that $\besttheta$ represents the true conditional
data-generating distribution.) The coefficient of the quadratic term
now becomes the conditional Fisher information at $\besttheta$ given
$X'$, so we have
\begin{multline*}
-2\E_{\mathbf Y' \mid X'}\log g(\mathbf Y' \mid X', \hat\theta(X, \mathbf Y))\\
= -2\E_{\mathbf Y' \mid X'}\log g(\mathbf Y' \mid X', \besttheta)\\
+ [\hat\theta(X, \mathbf Y) - \besttheta]\trans
  I(\besttheta \mid X')
  [\hat\theta(X, \mathbf Y) - \besttheta]
+r(\hat\theta).
\end{multline*}
Rearranging the quadratic term and taking the other expectation, we
obtain
\begin{multline}\label{eq:resStepI}
-2\E_{\mathbf Y \mid X}\E_{\mathbf Y' \mid X'}
  \log g(\mathbf Y' \mid X', \hat\theta(X, \mathbf Y))\\
=  -2\E_{\mathbf Y' \mid X'}\log g(\mathbf Y' \mid X', \besttheta)\\
+ \trace\left\{I(\besttheta \mid X')
\left[\E_{\mathbf Y \mid X}[\hat\theta(X, \mathbf Y) -
  \besttheta][\hat\theta(X, \mathbf Y) -
  \besttheta]\trans\right]\right\}\\
+\E_{\mathbf Y \mid X} r(\hat\theta).
\end{multline}
The other matrix in the trace is the conditional covariance matrix of
$\hat\theta(X, \mathbf Y)$.

To proceed with the first term on the right hand side, we use
Assumption~\ref{ass:FAIC}. Then we have
\begin{multline*}
-2\frac{n}{n'}\E_{\mathbf Y' \mid X'}\log g(\mathbf Y' \mid X', \besttheta)\\
= -2\E_{\mathbf Y \mid X}\log g(\mathbf Y \mid X, \besttheta)
\end{multline*}
for a sample $(X, \mathbf Y)$ of size $n$. (Here $X$ still represents
the values of the input variable in the training set, but $\mathbf Y$
conceptually represents a new sample.) Now only one $X$ remains, so
the rest of the derivation corresponds to that of standard AIC, which
gives us
\begin{multline}\label{eq:resStepII}
-2\E_{\mathbf Y \mid X}\log g(\mathbf Y \mid X, \besttheta)\\
= -2\E_{\mathbf Y \mid X}
  \log g(\mathbf Y \mid X, \hat\theta(X, \mathbf Y))
+ k + o(1).
\end{multline}
Multiplying \eqref{resStepI} by $n/n'$ and plugging in the above, we get
\begin{multline*}%
  -2\frac{n}{n'}\E_{\mathbf Y \mid X} \E_{\mathbf Y' \mid X'} \log
  g(\mathbf Y' \mid X', \hat\theta(X, \mathbf Y)) \\
= -2\E_{\mathbf Y \mid X}
  \log g(\mathbf Y \mid X, \hat\theta(X, \mathbf Y)) + k\\
  + \frac{n}{n'}\trace\left\{I(\besttheta \mid X')
    \Cov(\hat\theta(X, \mathbf Y) \mid X)\right\} \\
  + \E_{\mathbf Y \mid X} \frac{n}{n'}r(\hat\theta)
  + o(1).
\end{multline*}
The term with the trace is what we called $\kappa_{X'}$.

By the assumed asymptotic normality of the maximum likelihood
estimator, $\E_{\mathbf Y \mid X} n \lVert \hat\theta(X, \mathbf Y) -
\besttheta \rVert^2$ converges to a constant, so that the first
remainder term $\E_{\mathbf Y \mid X} (n/n')r(\hat\theta) =
(1/n')o(1)$; because we additionally assumed the test set is either
fixed or grows with the training set, this is again $o(1)$. This proves
\eqref{FAICthm}.

In the case of a linear model with fixed variance $\sigma^2$, the
second-order Taylor approximation and the approximation in
\eqref{resStepII} are actually exact.
\end{proof}

\begin{proof}[Proof of Theorem \ref{thm:FAICc}]
  This proof will follow a different path than the one above. It is
  adapted from the derivation of $\AICc$ in \citet[section
  7.4.1]{BurnhamAnderson}. We first consider the case where the
  training set size $n' = 1$. Then $X'$ becomes a vector (we choose to
  make it a column vector) and $\mathbf Y'$ a scalar; we write $x$ and
  $\mathbf y$ for these. Hats denote maximum likelihood estimates. For
  Gaussian densities, we get
\begin{multline*}
  T = -2\E_{\mathbf Y \mid X} \E_{\mathbf y \mid x}
  \log g(\mathbf y \mid x, \hat\theta(X, \mathbf Y))\\
  = \E_{\mathbf Y \mid X} \E_{\mathbf y \mid x} \bigg[ \log
    2\pi\hat\sigma^2(X, \mathbf Y)\\
    + \frac{1}{\hat\sigma^2(X, \mathbf
      Y)}\left(\mathbf y - x\trans
      \hat\mu(X, \mathbf Y)\right)^2 \bigg] \\
  = \E_{\mathbf Y \mid X} \log 2\pi\hat\sigma^2(X, \mathbf Y)\\
  + \E_{\mathbf Y \mid X} \frac{1}{\hat\sigma^2(X, \mathbf Y)}
  \E_{\mathbf y \mid x} \left(\mathbf y - x\trans \hat\mu(X, \mathbf
    Y)\right)^2.
\end{multline*}
We will call the final term $T'$. Writing $y\best$ for
$\E_{\mathbf y \mid x} y$ and $\sigma^2\best$ for $\mathbf y$'s
unknown variance, the inner expectation becomes
\begin{multline*}
  \E_{\mathbf y \mid x} \left(\mathbf y - x\trans \hat\mu(X, \mathbf Y)\right)^2 \\
  = \E_{\mathbf y \mid x} (\mathbf y - y\best)^2 + 2\left(y\best -
    x\trans \hat\mu(X, \mathbf Y)\right)\E_{\mathbf y \mid x}(\mathbf
  y - y\best)\\
  + \left(y\best - x\trans
    \hat\mu(X, \mathbf Y)\right)^2 \\
  = \sigma^2\best + x\trans(\mu\best-\hat\mu(X, \mathbf
  Y))(\mu\best-\hat\mu(X, \mathbf Y))\trans x.
\end{multline*}
Using the fact that $\hat\mu(X, \mathbf Y)$ and $\hat\sigma^2(X,
\mathbf Y)$ are independent in this
setting, %
\begin{multline*}
  T' = \left[\E_{\mathbf Y \mid X} \frac{1}{\hat\sigma^2(X, \mathbf
      Y)}\right] \\
  \cdot \left[ \sigma^2\best + x\trans \Cov(\hat\mu(X, \mathbf Y)
\mid X) x \right].
\end{multline*}
The covariance matrix equals $\sigma^2\best (X\trans X)^{-1}$. Then we
use that $n\hat\sigma^2/\sigma^2\best$ follows a chi-squared
distribution with $n-k+1$ degrees of freedom ($k$ is the number of
free parameters in the model, which includes $\sigma^2$), and that $\E
1/\chi^2_{n-k} = 1/(n-k-1)$:
\begin{multline*}
  T' = 
  \left[\E_{\mathbf Y \mid X} \frac{1}{\hat\sigma^2(X, \mathbf
      Y)}\right] \left[ \sigma^2\best
    + \sigma^2\best x\trans (X\trans X)^{-1} x \right] \\
  = 
  \left[\E_{\mathbf Y \mid X} \frac{\sigma^2\best}{n\hat\sigma^2(X,
      \mathbf Y)}\right] \left[ n
    + n x\trans (X\trans X)^{-1} x \right] \\
  = 
    \frac{n + n x\trans (X\trans X)^{-1} x}{n-k-1} \\
  = 1
  + \frac{n + n x\trans (X\trans X)^{-1} x - (n-k-1)}{n-k-1} \\
  = 1 +
  \frac{k + \kappa_x}{n-k-1},
\end{multline*}
where $\kappa_x = n x\trans (X\trans X)^{-1} x + 1$. The reason for
splitting off the $1$ from the fraction is that $n(\E_{\mathbf Y \mid
  X} \log2\pi\hat\sigma^2(X, \mathbf Y) + 1)$ is $-2$ times the maximized
log-likelihood. Then we multiply by $n$ and get the result in the
stated form:
\begin{multline*}
nT = -2\E_{\mathbf Y \mid X} \log g(\mathbf Y \mid X, \hat\theta(X, \mathbf Y))
    + \frac{n(k + \kappa_x)}{n-k-1} \\
 = -2\E_{\mathbf Y \mid X} \log g(\mathbf Y \mid X, \hat\theta(X, \mathbf Y))\\
    + k + \kappa_x + \frac{(k + 1)(k + \kappa_x)}{n-k-1}.
\end{multline*}

The result for $n' > 1$ now follows by taking the average over all
sample points in the test set on both sides.
\end{proof}

\begin{proof}[Proof of Proposition~\ref{prop:bias}]
  Assume without loss of generality that the variance is known (as its
  inclusion does not affect the statement of the theorem) and that the
  basis is orthonormal with respect to the measure underlying $\mathbf
  x$ (that is, that $\E_{\mathbf x} \mathbf x \mathbf x\trans = I_k$).
  Then
\begin{align*}
\E_{\mathbf x} \kappa_{\mathbf x} & = n \E_{\mathbf x} \mathbf x\trans (\mathbf
X\trans \mathbf X)^{-1} \mathbf x \\
  & = n \trace(\mathbf X\trans \mathbf X)^{-1}
  = \trace(\frac{1}{n} \mathbf X\trans \mathbf X)^{-1},
\end{align*}
where orthonormality was used in the second equality. To compare the
$\kappa_{\mathbf x}$ for this model with that of a submodel with one
fewer entry in its design vectors, write
\[
\frac{1}{n} \mathbf X\trans \mathbf X
= \begin{bmatrix}\mathbf A & \mathbf v \\ \mathbf v\trans & \mathbf d\end{bmatrix}.
\]
Note that by orthonormality, the expected value of this matrix is the
identity matrix. We require that its inverse exists. Then for
$\mathbf d'=(\mathbf d-\mathbf v\trans \mathbf A^{-1}\mathbf v)^{-1}$,
\begin{align*}
\E \kappa_{\mathbf x}
& = \E\trace\begin{bmatrix}\mathbf A & \mathbf v \\ \mathbf v\trans & \mathbf d\end{bmatrix}^{-1} \\
& = \E\trace\begin{bmatrix}
    \mathbf A^{-1}+\mathbf A^{-1} \mathbf v \mathbf d' \mathbf v\trans \mathbf A^{-1}
    & -\mathbf A^{-1} \mathbf v \mathbf d' \\
    -\mathbf d' \mathbf v\trans \mathbf A^{-1}
    & \mathbf d'
  \end{bmatrix} \\
& = \E\trace \mathbf A^{-1} + \E\frac{\mathbf v\trans \mathbf A^{-2} \mathbf v + 1}{\mathbf d-\mathbf v\trans \mathbf A^{-1}\mathbf v} \\
& \geq \E\trace \mathbf A^{-1} + \E\frac{1}{\mathbf d-\mathbf v\trans
  \mathbf A^{-1} \mathbf v} \\
& \geq \E\trace \mathbf A^{-1} + \frac{1}{1 - \E\mathbf  v\trans
  \mathbf A^{-1} \mathbf v} \\
& \geq \E\trace \mathbf A^{-1} + 1.
\end{align*}
This shows that adding an element to the design vector increases $\E
\kappa_{\mathbf x}$ by at least one. For $k=1$ (so that $\mathbf A$ is
a $0 \times 0$ matrix), we have equality if and only if $\mathbf d=1$
almost surely, which means that for ${\mathbf x}_1$ (the first
and only entry of design vector $\mathbf x$), we must have ${\mathbf
  x}_1 = \pm 1$ almost surely.
For $k \geq 2$, because
$\mathbf A^{-1}$ is positive definite, equality requires that $\mathbf
v$ is the zero vector almost surely (in addition to the same
requirement as above on all ${\mathbf x}_i$). But this can only be
satisfied if ${\mathbf x}_i {\mathbf x}_k = 0$ almost surely for all
$i<k$, which is incompatible with the conditions on ${\mathbf x}_1$
and ${\mathbf x}_k$.
\end{proof}
}{}

\end{document}